\documentclass[a4paper]{amsart}

\usepackage{amssymb}
\usepackage{amsthm}
\usepackage{enumitem}
\usepackage{graphicx}
\usepackage{algorithm2e}
\usepackage{bbm}
\usepackage{tikz}
\usepackage{url}
\usepackage{hyperref}
\begin{document}
\title[Optimal transport between DPP and fast simulation]{Optimal transport between determinantal point processes and application
  to fast simulation}

\author{Laurent Decreusefond}
\address{LTCI, Telecom Paris, Institut polytechnique de Paris, France}
\email{laurent.decreusefond@mines-telecom.fr}

\author{Guillaume Moroz}
\address{INRIA Nancy Grand-Est, Nancy, France}
\email{guillaume.moroz@inria.fr}

\thanks{Supported by grant ANR-17-CE40-0017 of the French National Research Agency (ANR project ASPAG)}

\begin{abstract}
  We analyze several optimal transportation problems between determinantal point
  processes. We show how to estimate some of the distances between distributions
  of DPP they induce. We then apply these results to evaluate the accuracy of a
  new and fast DPP simulation algorithm. We can now simulate in a reasonable
  amount of time more than ten thousands points.
\end{abstract}
\maketitle{}

%%% Environments
\newtheorem{theorem}{Theorem}
\newtheorem{corollary}[theorem]{Corollary}
\newtheorem{lemma}[theorem]{Lemma}
{\theoremstyle{definition}
\newtheorem{definition}{Definition}
\newtheorem{hyp}{Hypothesis}
\newtheorem{notation}{Notation}
}
{\theoremstyle{remark}
  \newtheorem{remark}{Remark}
}
%%% End of environments
%%% macros
\def\mug{\mu_{\text{Gin}}}%% Ginibre tronqué à R
\def\NN{{\mathfrak{N}}}
\def\proba{{\mathfrak M}_1}
\def\C{{\ensuremath{\mathbf C}}}
\def\dif{\text{d}}
\def\moins{\ominus}
\def\car{\mathbf 1}
\def\distrib{\text{law}}
\newcommand{\esp}[1]{\mathbf E\left[{#1}\right] }
\newcommand{\espp}[2]{\mathbf E_{#1}\left[{#2}\right] }
\def\trace{\operatorname{trace}}
\def\vect{\operatorname{span}}
\def\sp{\operatorname{sp}}
\def\proj{\operatorname{proj}}
\def\dist{\operatorname{dist}}
  \def\E{{\mathbf E}}
  \def\R{{\mathbf R}}
  \def\N{{\mathbf N}}
  \def\T{{\mathcal W}}
  \def\P{\mathbf P}
  \def\F{\mathcal F}
  \def\opt{{\text{opt}}}
  \def\c{{\mathfrak c}}
  \def\p{{\mathfrak p}}
  \def\Id{\operatorname{Id}}
  \def\Lip{\operatorname{Lip}}
  \def\MKP{\operatorname{MKP}}
\def\Adj{\operatorname{Adj}}
\def\TV{\text{\small{TV}}}
\def\KR{\text{\small{KR}}}
\def\Gin{\mathfrak G}
  %%% end of macros
%\input{introduction.tex}
\section{Introduction}
\label{sec:intro}

Determinantal point processes (DPP) have been introduced in the seventies
\cite{Macchicoincidenceapproachstochastic1975} to model fermionic particles with
repulsion like electrons. They recently regained interest since they represent
the locations of the eigenvalues of some random matrices. A determinantal point
process is characterized by an integral operator of kernel $K$ and a reference
measure~$m$. The integral operator is compact and symmetric and is thus characterized by
its  eigenfunctions and its eigenvalues. Following
\cite{HoughDeterminantalprocessesindependence2006}, the eigenvalues are not
measurable functions of the realizations of the point process so it is difficult
to devise how a modification of the eigenfunctions, respectively of the
eigenvalues or of the reference measure,   modifies the random configurations of
a DPP. Conversely, it is also puzzling to know how the usual transformations on point processes like
thinning, dilations, displacements, translate onto $K$ and  $m$.

A careful analysis of the simulation algorithm  given in
\cite{HoughDeterminantalprocessesindependence2006} yields several answers to
these questions. For instance, it is clear that the eigenvalues control the
distribution of the number of points and the eigenfunctions determine the
positions of the atoms once their number is known. The above mentioned algorithm
is a beautiful piece of work but requires to draw points according to
distributions whose densities  are not expressed as  combinations of classical
functions, hence the necessity to use rejection sampling method. Unfortunately,
as the number of drawn points increases, the densities quickly present high peaks and
deep valleys inducing a high number of rejections, see Figure~\ref{fig:peaks}.

\begin{figure}[!ht]
  \centering
  \includegraphics[width=0.45\textwidth]{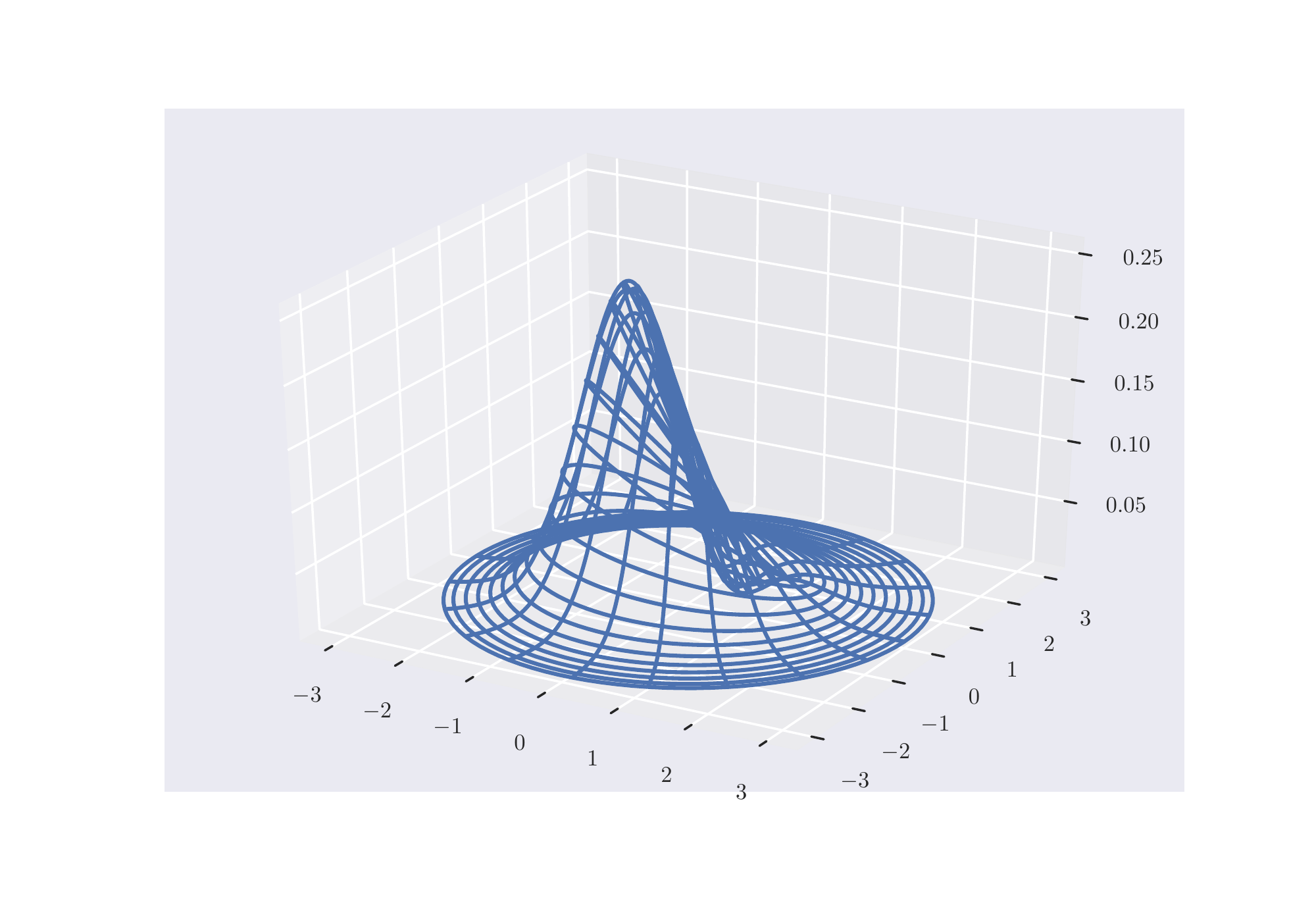}
  \includegraphics[width=0.45\textwidth]{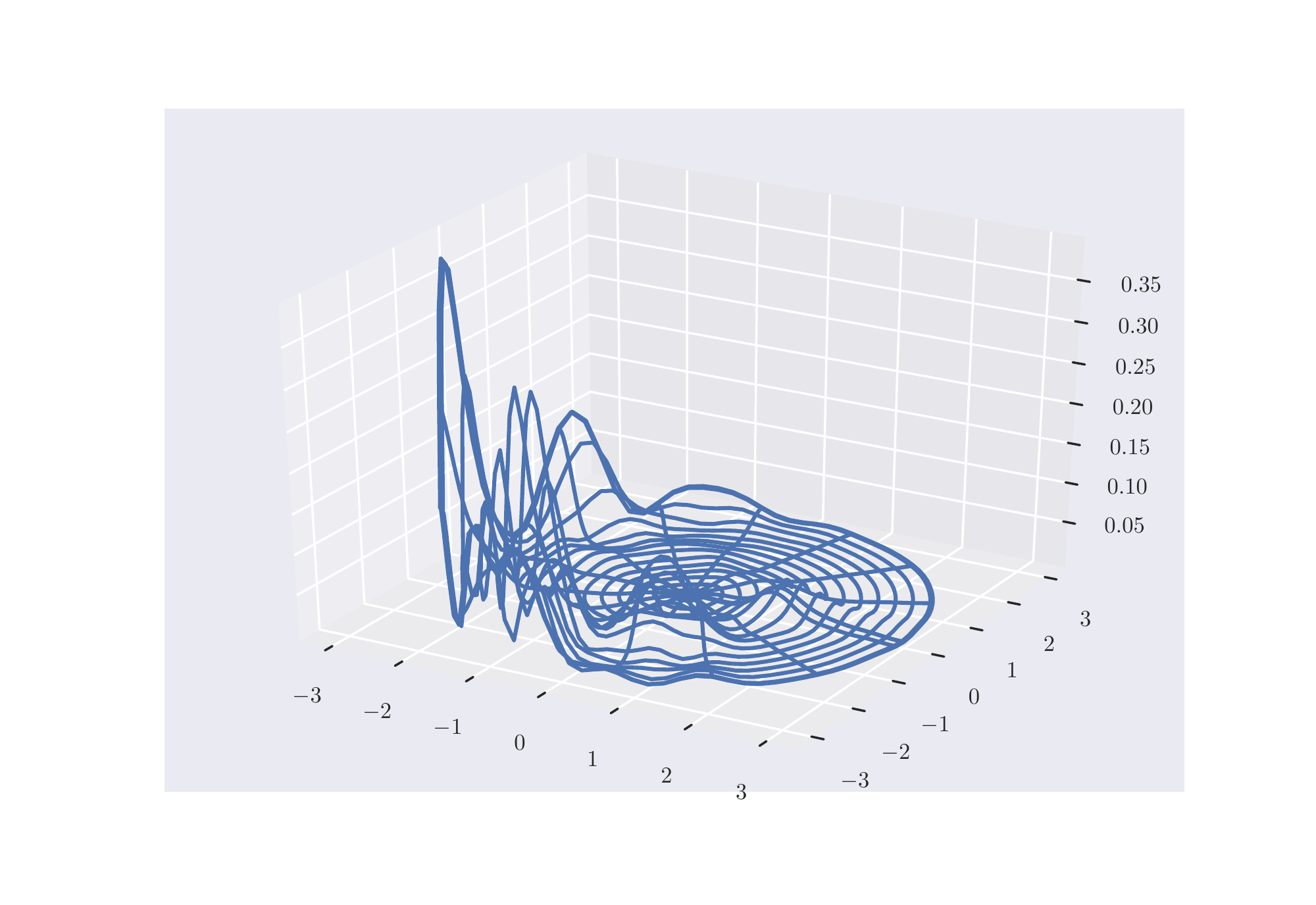}
  \caption{Peaks and valleys of some densities.}
  \label{fig:peaks}
\end{figure}

As a consequence, it is hardly feasible to simulate a DPP with more than one
thousand points in a reasonable amount of time. As a DPP appears as the locations of the eigenvalues
of some matrix ensembles, it may seem faster and simpler to draw random matrices
and compute the eigenvalues with the optimized libraries to do so. There are
several drawbacks to this approach: 1) we cannot control the domain into which
the points fall, for some applications it may be important to simulate DPP
restricted to some compact sets, 2) as eigenvalues belong to $\R$ or $\C$, we
cannot imagine DPP in higher dimensions with this approach, 3) for stationary
DPP, it is often useful to simulate under the Palm measure (see below) which is
known to correspond to the distribution of the initial DPP with the first
eigenvalue removed so no longer corresponds to a random matrix ensemble.

Several refinements of the algorithm~\ref{algo:sampling} have been proposed
along the years but the most advanced contributions have been made for DPP on
lattices, which are of a totally different nature than continuous DPPs.
 We here propose to fasten the simulation of a DPP by reducing the number of
 eigenvalues considered and approximating the eigenfunctions by functions whose
 quadrature can be easily inversed to get rid of the rejection part.

 We evaluate the impact of these approximations by bounding the distances
 between the original distribution of the DPP to be simulated and the real
 distribution according to which the points are drawn.

 Actually, there are several notions of distances between the distributions of
 point processes (see \cite{DecreusefondUpperboundsRubinstein2010} and
 references therein). We focus here on the total variation distance and on the
 quadratic Wasserstein distance. The former counts the difference of the number
 of points in an optimal coupling between two distributions. The latter
 evaluates the matching distance between two realizations of an optimal coupling provided that it exists.

The paper is organized as follows. We first recall the definition and salient
properties of DPP. In Section \ref{sec:optimal}, we briefly introduce the
optimal transportation problem in its full generality and give some elements
dedicated to point processes. In Section~\ref{sec:distance}, we show how the
eigenvalues and eigenfunctions do appear in the evaluation of the distances
under scrutiny. In Section~\ref{sec:simul}, we apply these results to the
simulation of DPPs.

\section{Determinantal point processes}
\label{sec:prel}
Let $E$ be a Polish space, $\mathcal{O}(E)$ the family of all non-empty open
subsets of $E$ and $\mathcal{B}$ denotes the corresponding Borel
$\sigma$-algebra. In the sequel, $m$ is a Radon measure on $(E,\mathcal{B})$. Let $\NN$ be
the space of locally finite subsets in $E$, also called the configuration space:
\begin{equation*}
  \NN = \{ \xi \subset E \,: \, | \Lambda \cap \xi | < \infty \, \text{ for any compact set } \Lambda \subset E    \},
\end{equation*}
equipped with the topology of the vague convergence. We call
elements of $\NN$ configurations and identify a locally finite configuration
$\xi$ with the atomic Radon measure $\sum_{y \in \xi} \varepsilon_y$, where we
have written $\varepsilon_y$ for the Dirac measure at $y \in E$.

Next, let $\NN^f = \{ \xi \in \NN \, : \, | \xi | < \infty \}$ the space of all
finite configurations on $E$. $\NN^f$ is naturally equipped with the trace
$\sigma$-algebra $\F^f = \F |_{\NN^f}$.

A random point process is defined as a probability measure on $(\NN, \F)$. A
random point process $\mu$ is characterized by its Laplace transform, which is
defined for any measurable non-negative function $f$ on $E$ as
\begin{equation*}
  {\mathcal L}_\mu (f) = \int_\NN e^{- \sum_{x \in \xi} f(x)} \dif\mu(\xi)	.
\end{equation*}
%
% processes.
Our notations are inspired by those of \cite{Georgii2005}, where the reader
can also find a brief summary of many properties of Papangelou intensities.
\begin{definition}
  We define the $m$-sample measure $L$ on $(\NN^f, \F^f)$ by the identity
  \begin{equation*}
    \int f(\alpha) \, \dif L(\alpha) = \sum_{n \ge 0} \frac{1}{n!} \int_{E^n} f(\{ x_1, \dots, x_n \}) \dif m(x_1) \ldots \dif m(x_n),
  \end{equation*}
  for any measurable nonnegative function $f$ on $\NN^f$.
\end{definition}
Point processes are often characterized via their correlation function defined as:
\begin{definition}[Correlation function]
  A point process $\mu$ is said to have a correlation function $\rho : \NN^f
  \rightarrow \R$ if $\rho$ is measurable and
  \begin{equation*}
    \int_\NN \sum_{\alpha \subset \xi, \ \alpha \in \NN^f} f ( \alpha) \dif \mu(\xi) = \int_{\NN^f} f ( \alpha )\, \rho(\alpha) \, \dif L(\alpha),
  \end{equation*}
  for all measurable nonnegative functions $f$ on $\NN^f$. For $\xi = \{ x_1,
  \dots, x_n \}$, we will write $\rho(\xi) = \rho_n(x_1,\dots,x_n)$
  and call $\rho_n$ the $n$-th correlation function, where $\rho_n$ is a
  symmetrical function on $E^n$.
\end{definition}
It can be noted that correlation functions can also be defined by the following
property, both characterizations being equivalent in the case of simple point
processes.
\begin{definition}
  A point process $\mu$ is said to have correlation functions $(\rho_n,\, n\ge 0)$ if for any $A_1,\dots,A_n$ disjoint bounded Borel subsets of $E$,
  \begin{equation*}
    \esp{\prod_{i=1}^n \xi(A_i)} = \int_{A_1\times \dots \times A_n} \rho_n(x_1,\dots,x_n) \dif m(x_1)\ldots \dif m( x_n)	.
  \end{equation*}
\end{definition}
Recall that $\rho_1$ is the mean density of particles with respect to $m$,
and
\begin{equation*}
  \rho_n(x_1,\dots,x_n)  \dif m(x_1)\ldots \dif m( x_n)
\end{equation*}
is the probability of finding a particle in the vicinity of each $x_i$,
$i=1,\dots,n$.

 Note
that
\begin{equation*}
  \NN_{E}^{f}=\bigcup_{n=0}^{\infty} \NN_{E}^{(n)}
\end{equation*}
where
\begin{equation*}
  \NN_{E}^{(n)}=\{\xi\in \NN_{E}^{f},\, \xi(E)=n\}.
\end{equation*}
Since $\NN_{E}^{(n)}$ can be identified with $E^{n}/\mathfrak S_{n}$ where
$\mathfrak S_{n}$ is the group of permutations over $n$ elements, every function $f\, :\, \NN_{E}^{f}\to \R$ is in fact equivalent to a family of
symmetric functions $(f_{n},\, n\ge 1)$ where $f_{n}$ goes from $E^{n}$ to $\R$.
For the sake of notations, we omit the index $n$ of $f_{n}$.
\begin{definition}
  A measure $\mu$ on $\NN_{E}^{f}$ is regular with respect to the reference
  measure $m$ when there exists
  \begin{align*}
    j \,:\, \bigcup_{n=0}^{\infty} \NN_{E}^{(n)}&\longrightarrow \R^{+}\\
    \{x_{1},\cdots,x_{n}\}&\longmapsto j_{n}(x_{1},\cdots,x_{n})
  \end{align*}
  where $j_{n}$ is symmetric on $E^{n}$ such that for
any measurable bounded  $f : \NN^{f} \rightarrow \R$,
\begin{equation}\label{eq_preliminaries:1}
  E[ f(\xi) ] = f(\emptyset)+ \sum_{n=1}^{\infty} \frac{1}{n!} \int_{E^n} f(x_1,\dots,x_n) j(x_1,\dots,x_n)\dif m(x_1)\dots \dif m(x_n).
\end{equation}
The function $j_{n}$ is called the $n$-th Janossy density. Intuitively, it can
be viewed as the probability to have exactly $n$ points in the vicinity
$(x_{1},\cdots, x_{n})$
\end{definition}
For details about the relationships between correlation functions and Janossy
densities, see  \cite{Daley2003}.

\subsection{Determinantal point processes}
\label{sec:dpp}

For details, we mainly refer to \cite{Shirai_2003}. A determinantal point on $X$ is characterized by a kernel $K$ and a reference
measure~$m$. The map $K$ is supposed to be an Hilbert-Schmidt operator from
$L^2(E,\, m)$ into $L^2(E,\, m)$ which satisfies the following conditions:
\begin{enumerate}
\item \label{item:1} $K$ is a bounded symmetric integral operator on $L^2(E,\,
  m)$, with kernel $K(.,.)$, i.e., for any $x\in E$,
  \begin{equation*}
    K f(x)=\int_{E}{K(x,y)f(y)\dif m(y)}.
  \end{equation*}
\item \label{item:2} The spectrum of $K$ is included in $[ 0,\, 1]$.
\item \label{item:3} The map $K$ is locally of trace class, i.e., for all
  compact $\Lambda\subset E$, the restriction $K_{\Lambda}=P_{\Lambda}K
  P_{\Lambda}$ of $K$ to $L^2(\Lambda,\, m)$
  is of trace class.
\end{enumerate}

\begin{definition}
  The determinantal measure on $\NN$ with characteristics $K$ and $m$ can be
  defined through its correlation functions:
  \begin{equation*}
    \rho_{n,\, K}(x_1,\cdots,\,\, x_n)=\det\bigl( K\left( x_k,\, x_l\bigr)
    \right)_{1\le k,l\le n},
  \end{equation*}
  and for $n=0$, $\rho_{0,\,K}(\emptyset)=1$.
  \end{definition}
There is a particular class of DPP which is the basic blocks on which general
DDP are built upon.
\begin{definition}
\label{def:projectionDPP}
  A DPP whose spectrum is reduced to the singleton $\{1\}$ is called a
  projection DPP. Actually, its kernel is of the form
  \begin{equation*}
    K_{\phi}(x,y)=\sum_{j=0}^{M} \phi_{j}(x)\phi_{j}(y)
  \end{equation*}
  where $M\in \N\cup\{\infty\}$ and  $(\phi_{j},\, j=0,\cdots, M)$ is a family
  of orthonormal functions of $L^{2}(E,m)$.

  If $M$ is finite then almost-all configurations of such a point process have
  $M$ atoms.
\end{definition}
Alternatively, when the spectrum of $K$ does not contain~$1$, we can define
  a DPP through its Janossy densities. In this situation, the properties of $K$ ensure that there exists a sequence $(\lambda_{i},\, i \ge
1)$ of elements of $[0,1)$ with no accumulation point but $0$ and a complete
orthonormal basis $(\phi_{i}, \, i\ge 1)$ of $L^{2}(m)$ such that
\begin{equation*}
  K_{\phi}(x,y)=\sum_{i\ge 1} \lambda_{i} \phi_{i}(x)\phi_{i}(y).
\end{equation*}
Note that if $L^{2}(E,\, m)$ is a $\C$-vector space, we must modify this definition
accordingly:
\begin{equation*}
  K_{\phi}(x,y)=\sum_{i\ge 1} \lambda_{i} \phi_{i}(x)\overline{\phi_{i}(y)}.
\end{equation*}

  For a compact subset $\Lambda\subset X$, the map $J_{\Lambda}$ is defined by:
\begin{equation*}
  J_{\Lambda} =\left(\Id-    K_{\Lambda}\right)^{-1} K_{\Lambda},
\end{equation*}
so that we have:
\begin{equation*}
  \left(\Id- K_{\Lambda} \right) \left(\Id+
    J_{\Lambda}\right) =\Id.
\end{equation*}

For any compact $\Lambda\subset
  E,$ the operator $J_{\Lambda}$ is an Hilbert-Schmidt, trace class
   operator, whose spectrum is included in $[0, +\infty)$. We denote by
   $J_{\Lambda}$ its kernel. For  any $n \in \N$,
  any compact $\Lambda\subset E$, and any $(x_1,\cdots,\, x_n) \in \Lambda^n$,
  the $n$-th Janossy density is given by:
  \begin{equation}\label{eq:4}
    j_{\Lambda}^{n}\left( x_1,\, \cdots,\, x_n\right)
    =\det\left( J_{\Lambda}\left( x_k,\, x_l\right)
    \right)_{1\le k,l\le n}.
  \end{equation}

We can now state how the characteristics of a DPP are modified by some usual
transformations on the configurations.
\begin{theorem}
  Let $\mu$ a DPP on $\R^{k}$ with kernel $K$ and reference measure $m=h\dif x$.
  Let $(\lambda_{n},\, n\ge 0)$ be its eigenvalues counted with multiplicity and
  $(\phi_{n}, \, n\ge 0)$ the corresponds eigenfunctions.

  \begin{enumerate}
  \item A random thinning of probability $p$ transforms $\mu$ into a DPP of
    kernel $pK$.
\item A dilation of ratio $\rho$ transforms $\mu$ into a DPP of kernel
  \begin{equation*}
   K_{\rho}(x,y)=\frac{1}{\rho}\, K(\rho^{-1/k}x,\, \rho^{-1/k}y).
  \end{equation*}
\item If $H$ is a $\mathcal C^{1}$-diffeomorphism on $E$, then
  \begin{align*}
    \mathcal H\, :\, \NN_{E}&\longrightarrow \NN_{E}\\
\sum_{x\in \xi}\varepsilon_{x}&\longrightarrow \sum_{x\in \xi}\varepsilon_{H(x)}
  \end{align*}
transforms $\mu$ into a DPP of kernel
\begin{equation*}
K_{H}(x,y)= K\left(H^{-1}(x),\, H^{-1}(y)\right)
\end{equation*}
and reference measure $m\circ H^{ -1}$, the image measure of $m$ by $H$, see
\cite{CamilierQuasiinvarianceintegrationparts2010}.
\item If $K(x,y)=k(x-y)$ then $\mu$ is stationary \cite{Lavancier2015} and thus admits a stationary
  Palm measure $\mu^{0}$. From \cite{Shirai_2003}, we know that $\mu^{0}$ is
  distributed as the DPP of kernel
  \begin{equation*}
    K^{0}(x,y)=\sum_{n=1}^{\infty} \lambda_{n}\phi_{n}(x)\phi_{n}(y).
  \end{equation*}
  \end{enumerate}
\end{theorem}
\begin{remark}
  It is straightforward to see
that the spectrum of $K_{H}$ in $L^{2}(E,\, m\circ H^{-1})$ is the same as the
spectrum of $K$ in $L^{2}(E,\, m)$. Actually, this transformation will be a
particular case of the optimal maps obtained in solving the $\MKP$ for the
Wassertein-2 distance (see Theorem~\ref{thm:constructionTransport}).
\end{remark}
\begin{remark}
  Recall that for a Poisson process, we can obtain a realization of its Palm
  measure by just adding an atom at~$0$ to any of its  realization.
  Unfortunately, we know from \cite{HoughDeterminantalprocessesindependence2006} that the
  eigenvalues of a DPP cannot be obtained as measurable functions of the configurations. Hence it is
  hopeless to construct a realization of $\mu^{0}$ from a realization of $\mu$.
\end{remark}

\subsubsection{Simulation of DPP}
\label{sec:simulDPP}

The simulation algorithm introduced
\cite{HoughDeterminantalprocessesindependence2006} is up to now the most
efficient to produce random configurations distributed according to a determinantal point
process. It is based on the following lemma.
\begin{lemma}
  \label{lem:eigenvalues}
  Let $\mu_{K,m}$ a determinantal point process of a trace-class kernel $K$ and reference
  measure $m$. Let $\sp(K;L^{2}(m))=\{\lambda_{n}, i\ge 0\}$ and $(\phi_{n},\,
  i\ge 0)$ a CONB of $L^2(E,\,m)$ composed of eigenfunctions of $K$. Let
  $(B(\lambda_{n}),i\ge 0)$ a family of independent Bernoulli random variables
  of respective parameter $\lambda_{n}$. Let
  \begin{equation*}
    I=\{n\ge 0, \, B(\lambda_{n})=1\}.
  \end{equation*}
  Since $\esp{|I|}=\sum_{n=0}^{\infty}\lambda_{n}<\infty$, $I$ is a.s. a finite subset of
  $\N$. Consider
  \begin{equation*}
    K_{I}(x,y)=\sum_{n\in I} \phi_{n}(x)\phi_{n}(y)
  \end{equation*}
  and
  \begin{equation*}
    p_{I}(x_{1},\cdots,x_{|I|})=\frac{1}{|I|!}\det\Bigl( K_{I}(x_{k},x_{l}),\, 1\le k,l\le |I| \Bigr).
  \end{equation*}
  Construct a random configuration  $\xi$ as follows: Given $I$, draw
  points $(W_{1},\cdots,W_{|I|})$ with joint density $p_{I}$. Then $\xi$ is
  distributed according to $\mu_{K,m}$.
\end{lemma}

In the following, let $\phi_I(x)=(\phi_n(x),n\in I)$. % The costly part is the
% drawing of points according to $p_{I}$.

 \begin{algorithm}[H]
     \KwData{$R,I$}
  \KwResult{$W_1,\cdots,W_{|I|}$}
  \BlankLine
  Draw $W_1$ from the distribution with density $\|\phi_I(x)\|_{\C^{|I|}}^2/|I|$\;
  $e_1\leftarrow \phi_I(W_1)/\|\phi_I(W_1)\|_{\C^{|I|}}$\;
  \For{$i\leftarrow 2$ \KwTo $|I|$}
    {Draw $W_i$ from the distribution with density
      \begin{equation*}
  p_{i}(x)=      \frac{1}{|I|-i+1}\left(\left\|\phi_I(x)\right\|_{\C^{|I|}}^2-\sum_{k=1}^{i-1}|e_k.\phi_I(x)|^2\right)
      \end{equation*}
      \noindent $u_i\leftarrow \phi_I(W_i)-\sum_{k=1}^{i-1}e_k.\phi_I(W_i)\ e_k$\;
      $e_{i}\leftarrow u_i/\|u_i\|_{\C^{|I|}}$\;
      }
  \caption{Sampling of the locations of the points given the set $I$ of active Bernoulli random variables.}
  \label{algo:sampling}
\end{algorithm}

\medskip

We have two kind of difficulties here: the drawing of $W_{i}$ according to a
density function with no particular feature so we usually have to  resort to rejection
sampling; when $|I|$ is large the computation of the density may be costly as it
contains a sum of $|I|$ terms. Figure~\ref{fig:peaks} also suggests that when the number of
points becomes high, the profile of the conditional density might be very
chaotic with high peaks and deep valleys, involving a large number of rejections
in the sampling of this density. These are the problems we intend to address in
the following.

\begin{remark}
Note that this algorithm is fully applicable even if $E$ is a discrete finite
space. It has been improved in several ways \cite{MAL-044,2018arXiv180208471T}
but when it comes to simulate a DPP with a large number of points as it is
necessary in some applications \cite{Bardenet2016}, the best way
remains to use MCMC methods \cite{pmlr-v49-anari16}. Unfortunately, by its very
construction, this last approach is not feasible when the underlying space $E$ is continuous.

\end{remark}

\section{Distances derived from optimal transport}
\label{sec:optimal}
For details on optimal transport in $\R^{d}$ and in general Polish spaces, we
refer to \cite{VillaniOptimaltransportold2007,Villani2003}. For $X$ and $Y$ two
Polish spaces, for $\mu$ (respectively $\nu$) a probability measure on $X$
(respectively $Y$), $\Sigma(\mu, \, \nu)$ is the set of probability measures on
$X\times Y$ whose first marginal is $\mu$ and second marginal is $\nu$. We also
need to consider a lower semi continuous function $c$ from $X\times Y$ to
$\R^+$. The Monge-Kantorovitch problem associated to $\mu$, $\nu$ and $c$,
denoted by $\MKP(\mu$, $\nu$, $c$) for short, consists in finding
\begin{equation}
  \label{eq:3}
  \inf_{\gamma\in \Sigma(\mu, \, \nu)}\int_{X\times Y}c(x,y)\dif \gamma(x,\, y).
\end{equation}
More precisely, since $X$ and $Y$ are Polish and $c$ is l.s.c., it is known from
the general theory of optimal transportation, that there exists an optimal
measure $\gamma\in \Sigma(\mu,\, \nu)$ and that the minimum coincides with
\begin{equation*}
  \sup_{(F,\, G)\in \Phi_c}(\int_X F\dif \mu+\int_Y G\dif \nu),
\end{equation*}
where $(F,\, G)$ are such that $F\in L^1(\dif \mu)$, $G\in L^1(\dif \nu)$ and
$F(x)+G(y)\le c(x,\, y)$. We will denote by $\T_c(\mu, \, \nu)$ the value of the
infimum in \eqref{eq:3}.
In the sequel, we  need the following
theorem of Brenier:
\begin{theorem}
  \label{thm:optimal-transport}
  Let $c(x,y)=2^{-1}\|x-y\|^{2}$ be the Euclidean distance on $\R^{k}$ and $\mu,\, \nu$ two probability
  measures with finite second moment. If the measure $\mu$ is absolutely
  continuous with respect to the Lebesgue measure, there exists a unique optimal
  measure $\gamma_{\opt}$ which realizes the minimum in \eqref{eq:3}. Moreover,
  there exists a unique function $\psi\,:\, \R^{k}\to \R$ such that
  \begin{equation*}
    y=x - \nabla \psi(x), \ \gamma_{\opt} \text{-a.s.}
  \end{equation*}
  Then, we have
  \begin{equation*}
    \T_e(\mu,\nu)=\frac12\int_{\R^{k}}\|\nabla \psi\|_{\R^{k}}^{2}\dif \mu.
  \end{equation*}
  The square root of $\T_e(\mu,\nu)$ defines a distance on $\proba(\R^{k})$, the
  set of probability measures on $\R^{k}$,
  called the Wasserstein-2 distance.
\end{theorem}
For $c$ a distance on $X=Y$, $\T_{c}$ also defines a distance on
$\proba(\R^{k})$, often called Kantorovitch-Rubinstein or Wasserstein-1
distance. It admits the alternative  characterization.
\begin{theorem}[See \cite{DudleyRealanalysisprobability2002}]
  \label{thm:KR}Let $c$ be a
  distance on the Polish space $(X,d_{X})$.
  For $\mu$ and $\mu $ two probability measures on $X$,
  \begin{equation*}
    \T_{c}(\mu,\mu)=\sup_{f\in \Lip_{1}}\left( \int_{X}f\dif \mu-\int_{X}f\dif \nu \right)
  \end{equation*}
  where
  \begin{equation*}
    \Lip_{1}=\left\{ f\, :X\to \R, \ \forall x,y\in X, |f(x)-f(y)|\le d_{X}(x,y) \right\}.
  \end{equation*}
Note that $d_{X}$ is the distance  which defines the topology of the Polish
space $X$, it  may not be  identical to $c$.
\end{theorem}
The next result is found in \cite[Chapter 7]{Villani2003}.
\begin{theorem}
\label{thm:strongerTopology}The topologies induced by $\T_{c}$ and $\T_{e}$ on $\proba(\R^{k})$ are
strictly stronger than the topology of convergence in distribution.
\end{theorem}

\section{Distances between point processes}
\label{sec:distance}
There are several ways to define a distance between point processes. We here
focus on two of them. They are constructed similarly: Choose a cost function $c$ on
$\NN_{E}$ and then consider $\T_{c}$ defined by the solution of $\MKP(\mu,\nu,c)$
for $\mu$ and $\nu $ two elements of $\proba(\NN_{E})$.

\begin{definition}
  Consider $\dist_{\TV}$ the distance in total variation between two configurations (viewed as
  discrete measures):
  \begin{equation*}
    \dist_{\TV}(\xi,\zeta)=(\xi\Delta \zeta)(E)
  \end{equation*}
  where $\xi\Delta \zeta$ is the symmetric difference between the two sets $\xi$
  and $\zeta$, i.e. we count the number of distinct atoms between $\xi$ and
  $\zeta$.   Then, for $\mu$ and $\nu $ belonging to $\proba(\NN_{E})$, their
  Kantorovitch-Rubinstein distance is defined by
  \begin{align}
    \T_{\KR}(\mu,\nu)& = \inf_{\substack{\distrib(\xi)=\mu\\ \distrib(\zeta)=\nu}}\esp{(\xi\Delta \zeta)(E)}\notag\\
    &=\sup_{f\in \Lip_{1}(\NN_{E})}\left( \int_{\NN_{E}}f(\xi)\dif\mu(\xi)- \int_{\NN_{E}}f(\zeta)\dif\nu(\zeta) \right).\label{eq:1}
  \end{align}
  \end{definition}
  \begin{remark}
    For any compact set $\Lambda\subset E$, the map
    \begin{align*}
      \Xi_{\Lambda}\, :\, \NN_{E}&\longrightarrow \N\\
      \xi&\longmapsto \xi(\Lambda)
    \end{align*}
    is Lipschitz. Let $(\mu_{n},\, n\ge 1)$ be a sequence of point processes and denote by $\xi_{n}$ an
    $\NN_{E}$-valued random variable whose distribution is $\mu_{n}$. Similarly,
    for another element $\nu\in \proba(\NN_{E})$, let $\zeta$ be an
    $\NN_{E}$-valued random variable whose distribution is $\nu$.
    In view of \eqref{eq:1} and Theorem~\ref{thm:KR}, if $\T_{\KR}(\mu_{n},\nu)$ tends to zero then for any
    compact set $\Lambda$, the sequence of random variables
    $(\xi_{n}(\Lambda),\, n\ge 1)$
    converges in distribution to~$\zeta(\Lambda)$.
  \end{remark}

For the quadratic distance, we first consider, on $E=\R^k$, the cost function as $\rho(x,y)=2^{-1}\|x-y\|^2$ and we  define a cost between configurations (see also
\cite{BarbourCompoundPoissonprocess2002,BarbourCompoundPoissonapproximation2000,barbour_stein-chen_1992}) as the 'lifting' of $\rho$
on $\NN_E$:
\begin{equation*}
%\label{eq:definition_du_cout}
  c(\xi_1,\xi_2)=\inf \left\{ \int \rho(x,y)\  \dif\beta(x,y),\
    \beta\in \Sigma({\xi_1,\xi_2})\right\},
\end{equation*}
where $ \Sigma({\xi_1,\xi_2})$ denotes the set of
$\beta\in\NN_{E\times E}$ having marginals $\xi_1$ and $\xi_2.$  First remark
that when $\xi_{1}(E)$ is finite,   the cost is finite only if
$\xi_{1}(E)=\xi_{2}(E)$, otherwise $\Sigma({\xi_1,\xi_2})$ is empty and then, by
convention, the
cost is infinite. Moreover,  the cost is attained at  the permutation of
$\{1,\cdots, \xi_{1}(E)\}$ which minimizes the sum of the squared distances:
\begin{equation*}
  c(\xi_1,\xi_2)=\frac12\min_{\sigma\in \mathfrak S_{\xi_{1}(E)}} \sum_{j=1}^{\xi_{1}(E)}\|x_{i}-y_{\sigma(i)}\|^{2}
\end{equation*}
where $\xi_{1}=\{x_{j},1\le j\le \xi_{1}(E)\}$ and $\xi_{2}=\{y_{j},1\le j\le \xi_{1}(E)\}$.
For infinite configurations, it is not immediate that the cost function so defined
has the minimum regularity required to consider an optimal transport problem.
According to \cite{RocknerRademachertheoremconfiguration1999}, this is indeed
true as  $c$ is lower semi continuous on
$\NN_{E}\times\NN_E$.
We can then consider the Monge-Kantorovitch
problem $\MKP(\mu,\nu,c)$ on $\proba(\NN_{E})$.
The main theorem of \cite{DecreusefondWassersteindistanceconfigurations2008} is
the following. For $\Lambda$ a compact subset of $E$, by definition of locally
finite point process, the number of points of
$\xi_{|\Lambda}$ is finite hence we can write
\begin{equation*}
  \NN_{\Lambda }=\bigcup_{n=0}^{\infty} \NN_{\Lambda}^{(n)}
\end{equation*}
where
\begin{equation*}
  \NN_{\Lambda}^{(n)}=\left\{ \xi\in \NN_{\Lambda},\ \xi(\Lambda)=n \right\}.
\end{equation*}
\begin{definition}
A probability measure $\mu$ on $\NN_{E}$ is said to be regular whenever it
admits Janossy densities of any order.
\end{definition}

\begin{theorem}
  \label{thm:transport_fpp}
  Let $\Lambda\subset E$ a compact set.  Let $\mu$  be a regular probability measure on $\NN_\Lambda$ and
  $\nu$ be a probability measure on $\NN_E$. The Monge-Kantorovitch
  distance, associated to $c$, between $\mu$ and $\nu$ is finite if
  and only if the following two conditions hold
  \begin{enumerate}
  \item \label{item:4} $\mu(\zeta(\Lambda)=n)=\nu(\xi(E)=n):=\c_{n}$ for any integer
  $n\ge 0$,
\item  \label{item:5} $\sum_{n\ge 1}\c_{n}\,\T_{e}(\mu_n,\, \nu_n)$ is finite.
  \end{enumerate}
Then, the solution of $\MKP(\mu,\nu,c)$ is attained at a unique point $\gamma_{\opt}$ and there
exists a unique map
\begin{align*}
  \varphi\, :\, \bigcup_{n=0}^{\infty} \NN_{\Lambda}^{(n)}&\longrightarrow \bigcup_{n=0}^{\infty} \NN_{E}^{(n)}\\
 \xi= \{x_{1},\cdots,x_{n}\}&\longmapsto \{\varphi_{n}\bigl(y;\,(x_{1},\cdots,x_{n})\bigr),\ y\in\xi \}\in E^{n}
\end{align*}
such that for $\xi\in \NN_{\Lambda}^{(n)},$
\begin{equation*}
  \zeta=\sum_{x\in \xi} \varepsilon_{\varphi_{n}(x,\, \xi)}, \gamma_{\opt}\text{-a.s.}
\end{equation*}
Moreover,
\begin{equation}
\label{eq:8}
  \T_c(\mu,\, \nu)=\sum_{n\ge 1}\c_{n}\T_{e}(\mu_n,\, \nu_n).
\end{equation}
If $\nu$ is regular, then for $y\in \{x_{1},\cdots,x_{n}\}$,
\begin{equation*}
  \varphi_{n}(y,(x_{1},\cdots,x_{n}))=y-\nabla_{y} \psi_{n}(x_{1},\cdots,x_{n})
\end{equation*}
where $\Id-\nabla \psi_{n}$ is the optimal transportation map between the
$n$-th Janossy normalized  densities  $\c_{n}^{-1}\,j_{n}^{\mu}$ and $\c_{n}^{-1}\,j_{n}^{\nu}$.
\end{theorem}
This means that whenever the distance between $\mu$ and $\nu$ is finite, there
exists a strong coupling which works as follows: 1) draw a discrete random
variable with the distribution of $\xi(\Lambda)$, let $\iota$ the obtained value 2) draw the points of $\xi$
according to $\mu_{\iota}$ and then 3) apply the map $\varphi_{\iota}(.,\, \xi)$ to each
point of $\xi$. The configuration which is obtained is distributed according to $\nu_{\iota}$.

It is shown in \cite{DecreusefondWassersteindistanceconfigurations2008} that for
two Poisson point processes of respective intensity $\sigma_{1}$ and
$\sigma_{2}$, the distance defined above is finite if and only if
$\sigma_{1}(E)=\sigma_{2}$ and
\begin{equation*}
   \zeta= \sum_{x\in \xi}\varepsilon_{t(x)}, \ \gamma_{\opt} \text{-a.s.}
\end{equation*}
where $t$ is the optimal transport map between $\sigma_{1}$ and $\sigma_2$ for
the Euclidean cost as defined in Theorem \ref{thm:optimal-transport}. Note that
the optimal map is a transformation which is applied to each atom irrespectively
of the others. In full generality, for non Poisson processes,  the amount by which an atom is moved depends on the
other locations.

\subsection{Distances between DPP}
\label{sec:distancesDPP}

For determinantal point processes, we can evaluate the effect of a modification
of the eigenvalues with the Kantorovitch-Rubinstein distance and the effect of a
modification of the eigenvectors with the Wasserstein-2 distance.
\begin{lemma}
  \label{thm_transport:distanceBetweenProjectionKernel}
  Let $\mu$ and $\nu$ two determinantal point processes with respective kernels
  $K_{\mu}$ and $K_{\nu}$. Assume that $K_{\mu}$ and $K_{\nu}$ are two
  projection kernels in some Hilbert space $L^{2}(m)$ such that $K_{\mu}= K_{\nu}+L $ where $L$ is another
  projection kernel and $L$ is orthogonal to $K_{\nu}$. Then,
  \begin{equation}
    \label{eq:distanceProjective}
    \T_{\KR}(\mu,\,\nu)\le \text{range}(L).
  \end{equation}
\end{lemma}
\begin{proof}
  The hypothesis means that there exists $(\phi_{j},\, j=1,\cdots,l+n)$ a family
  of orthonormal functions in $L^{2}(m)$ such that
  \begin{equation*}
    K_{\nu}(x,y)=\sum_{j=1}^{n} \phi_{j}(x)\phi_{j}(y) \text{ and } L(x,y)=\sum_{j=n+1}^{l} \phi_{j}(x)\phi_{j}(y).
  \end{equation*}
  Since $L$ is a positive symmetric operator, this exactly  means that $K_{\nu}\prec
  K_{\mu}$ in the Loewner sense. According to \cite{goldman_palm_2010}, there
  exists $\xi',\, \zeta'$  of respective distribution $\mu$, $\nu$
  and a point process $\omega'$ such that
  \begin{equation*}
    \xi'=\zeta'+\omega' \text{ and } \zeta'\cap \omega'=\emptyset.
  \end{equation*}
  This implies that
  \begin{equation*}
    \xi'\Delta\zeta'(E)=\omega'(E)=l.
  \end{equation*}
According to the first definition of $\T_{\KR}$, see  \eqref{eq:1}, this implies \eqref{eq:distanceProjective}.
\end{proof}

\begin{theorem}\label{thm:distanceTV}
   Let $\mu$ (respectively $\nu$) be a determinantal point process of
  characteristics $K_{\mu}$ and $h_{\mu}$ (respectively $K_{\nu}$ and
  $h_{\nu}$) on a compact set $\Lambda\subset \R^{k}$. Denote by
  $(\lambda_{n}^{\mu},\, n\ge 0)$ (respectively $(\lambda_{n}^{\nu},\, n\ge 0)$)
  the eigenvalues of $K_{\mu}$ in $L^{2}(E,\,h_{\mu}\dif x)$ (respectively of
  $K_{\nu}$ in $L^{2}(E,\,h_{\nu}\dif x)$) counted with multiplicity and ranked in
  decreasing order. Assume that
  \begin{equation*}
    \lambda_{n}^{\nu} \le \lambda_{n}^{\mu} , \ \forall n\ge 0.
  \end{equation*}
  Then,
  \begin{equation}
    \label{eq:distanceTV}
    \T_{\KR}(\mu,\, \nu)\le \sum_{n=0}^{\infty} |\lambda_{n}^{\mu}-\lambda_{n}^{\nu}|.
  \end{equation}
\end{theorem}
\begin{proof}
We make a coupling of $(B(\lambda_{n}^{\mu}),\, n\ge 0)$ and
$(B(\lambda_{n}^{\mu}),\, n\ge 0)$ by using the same sequence of uniform random
variables: Let $(U_{n},\, n\ge 0)$ be a sequence of independent, identically uniformly distributed
over $[0,1]$, random variables, consider
\begin{equation*}
  X^{\mu}_{n}=\car_{\{U_{n}\le \lambda_{n}^{\mu}\}} \text{ and } X_{n}^{\nu}=\car_{\{U_{n}\le \lambda_{n}^{\nu}\}}.
\end{equation*}
Note that
\begin{equation}\label{eq_transport:2}
  \P(X_{n}^{\nu} \neq X_{n}^{\mu})= |\lambda_{n}^{\mu}-\lambda_{n}^{\nu}|.
\end{equation}
Let $I_{\mu}=\{n\ge 0, X_{n}^{\mu}=1\}$ and $I_{\nu}=\{n\ge 0, X_{n}^{\nu}=1\}$.
In view of the hypothesis, $X^{\nu}\le X^{\mu}$ hence $I_{\nu}\subset I_{\mu}$.
Otherwise stated, $K_{I_{\mu}}$ and $K_{I_{\nu}}$ are two projection kernels
which satisfy the hypothesis of
Lemma~\ref{thm_transport:distanceBetweenProjectionKernel}. Hence, there exists a realization $(\xi,\zeta)$ of $\Sigma(\mu,\nu)$ given $I_{\mu}$
and $I_{\nu}$,  such that
\begin{equation*}
\dist_{\TV}(\xi,\, \zeta)={\sum_{n=0}^{\infty }\car_{\{X_{n}^{\nu} \neq X_{n}^{\mu}\}}}.
\end{equation*}
Gluing these realizations together, we get a coupling  $(\xi,\zeta)$ such that
\begin{align*}
  \esp{\dist_{\TV}(\xi,\, \zeta)}&=\esp{\esp{\dist_{\TV}(\xi,\, \zeta)\,|\,I_{\mu},I_{\nu} }}\\
  &                                =\esp{\sum_{n=0}^{\infty }\car_{\{X_{n}^{\nu} \neq X_{n}^{\mu}\}}}\\
   &=\sum_{n=0}^{\infty } |\lambda_{n}^{\mu}-\lambda_{n}^{\nu}|,
\end{align*}
according to~\eqref{eq_transport:2}. Since the Kantorovitch-Rubinstein distance
is obtained as the infimum over all couplings of the total variation distance
between $\xi$ and $\zeta$,
this particular construction shows that \eqref{eq:distanceTV} holds.
\end{proof}
The next corollary is an immediate consequence of the alternative definition of
the KR distance on point processes, see Eqn. \eqref{eq:1}.
\begin{corollary}
  With the hypothesis of Theorem~\ref{thm:distanceTV}, let $\xi$ and $\zeta$ be
  random point process of respective distribution $\mu$ and $\nu$. Then, we have that
  \begin{equation*}
    \sup_{A\subset\Lambda} \dist_{\TV}(\xi(A),\, \zeta(A))\le \sum_{n=0}^{\infty} |\lambda_{n}^{\mu}-\lambda_{n}^{\nu}|.
  \end{equation*}
\end{corollary}
This means that the Kantorovitch-Rubinstein distance between point processes
focuses on the number of atoms in any compact. As we shall see now, the Wasserstein-2 distance
evaluates the matching distance between configurations  when they have the same cardinality.
\begin{theorem}
  \label{thm:egalite_vp}
  Let $\mu$ (respectively $\nu$) be a determinantal point process of
  characteristics $K_{\mu}$ and $h_{\mu}$ (respectively $K_{\nu}$ and
  $h_{\nu}$) on a compact set $\Lambda\subset \R^{k}$. The Wasserstein-2 distance between $\mu$ and $\nu$ is finite if
  and only if
  \begin{equation*}
    \text{sp}(K_{\mu};\, L^{2}(h_{\mu}\!\dif x))= \text{sp}(K_{\nu};\, L^{2}(h_{\nu}\!\dif x)).
  \end{equation*}
\end{theorem}
\begin{proof}[Proof of Theorem~\protect{\ref{thm:egalite_vp}}]
  According to point~\ref{item:4} of Theorem~\ref{thm:transport_fpp}, we must
  first prove the equality of the spectra. We already know that the eigenvalues
  of both kernels are between $0$ and $1$, with no other accumulation point than
  $0$. Furthermore, the distribution of $\zeta(\Lambda)$ is that of the sum of
  independent Bernoulli random variables of parameters given by the eigenvalues,
  hence
  \begin{equation}
    \label{eq:generatingFunction}
    \Phi_{\mu}(z)=\espp{\mu}{z^{\zeta(\Lambda)}}=\prod_{\lambda\in \sp K_\mu} (1-\lambda+\lambda z).
  \end{equation}
  The infinite product is convergent since $\trace K_{\mu}=\sum_{\lambda\in \sp
    K_\mu}$ is finite.

  If the Wasserstein-2 distance between $\mu$ and $\nu$ is finite then
  $\Phi_{\mu}=\Phi_{\nu}$. The zeros of these two holomorphic functions  are all greater
  than $1$ and are isolated. Let
  \begin{equation*}
    m(\Phi,r)=\text{number of zeros (counted with multiplicity) of $\Phi$ in $B(0,\, r)$}.
  \end{equation*}
  By the properties of zeros of holomorphic functions we have
  \begin{equation*}
    m(\Phi_{\mu},r)=m(\Phi_{\nu},r) \text{ for any }r\ge 0.
  \end{equation*}
  Hence,
  \begin{equation*}
    \Big\{\frac{1-\lambda}{\lambda}, \ \lambda \in \sp K_{\mu}\Big\}= \Big\{\frac{1-\lambda}{\lambda}, \ \lambda \in \sp K_{\nu}\Big\}
  \end{equation*}
  and the two spectra must coincide. Now then, by the very definition of
  $\T_{e}$,
  \begin{equation*}
    \T_{e}(\mu_{n},\nu_{n})\le \int_{\Lambda} \|x\|^{2}(\dif \mu_{n}+\dif \nu_{n})\le \sup_{x\in \Lambda}\|x\|^{2} \Bigl( \mu_{n}(\Lambda)+\nu_{n}(\Lambda) \Bigr).
  \end{equation*}
  Thus, we have
  \begin{multline*}
    \sum_{n\ge 1}\T_{e}(\mu_n,\, \nu_n)\ \mu(\zeta(\Lambda)=n)\le  \sup_{x\in
      \Lambda}\|x\|^{2}  \sum_{n\ge 1}\Bigl( \mu_{n}(\Lambda)+\nu_{n}(\Lambda)
    \Bigr) \mu(\zeta(\Lambda)=n)\\ = 2\sup_{x\in \Lambda}\|x\|^{2} \trace K_{\mu}.
  \end{multline*}
This quantity is finite hence the Wasserstein-2 distance between $\mu$ and $\nu$
as soon as the spectra are equal.
\end{proof}
The next lemma is a straightforward consequence of Lemma~\ref{lem:eigenvalues}.
\begin{lemma}
  \label{lem:representationJanossy}
  Let $\mu$ be a determinantal point process of
  characteristics $K_{\mu}$ and $h_{\mu}$. For $I$ a finite subset, let
  \begin{equation*}
    c_{I}=\prod_{i\in I}\lambda_{i}^{\mu}\, \prod_{j\in I^{c}}(1-\lambda_{j}^{\mu}),
  \end{equation*}
  where the $\lambda^{\mu}_{i}$'s are the eigenvalues of $K_{\mu}$. Then, its $n$-th Janossy density is
  given by
  \begin{equation*}
    j_{n}^{\mu}(x_{1},\cdots,x_{n})=\sum_{\substack{I\subset \N\\ |I|=n}}c_{I}\, p_{I}(x_{1},\cdots,x_{n}).
  \end{equation*}
  This means that given $\zeta(E)=n$, the points are distributed according to the
  probability measure:
  \begin{equation*}
   \p^{\mu}_{n}\,:\, (x_{1},\cdots,x_{n})\longmapsto\c_{n}^{-1}\sum_{\substack{I\subset \N\\ |I|=n}}c_{I}\, p_{I}(x_{1},\cdots,x_{n}) \text{ where } \c_{n}=\sum_{\substack{I\subset \N\\ |I|=n}}c_{I}.
 \end{equation*}
\end{lemma}
\begin{proof} Consider that $\xi$ is constructed with
  Algorithm~\ref{algo:sampling} and denote by $I_{\xi}$ the set of indices of
  the Bernoulli
  random variables which are equal to $1$ for the drawing of $\xi$.
  For any bounded $f\, :\, \NN^{f}_{E}\to \R$,
  \begin{align*}
    \esp{f(\xi)}&=f(\emptyset)+\sum_{n=1}^{\infty} \esp{f(\xi)\car_{\{\xi(E)=n\}}}\\
                &=f(\emptyset)+\sum_{n=1}^{\infty}\sum_{\substack{J\subset \N\\|J|=n}} \esp{f(\xi)\,|\,{I_{\xi}=J}}\, c_{J}\\
                &=f(\emptyset)+\sum_{n=1}^{\infty}\sum_{\substack{J\subset \N\\|J|=n}} c_{J}\int_{E^{n}} f(x_{1},\cdots,x_{n})\, p_{J}(x_{1},\cdots, x_{n})\dif x_{1}\ldots \dif x_{n}\\
    &=f(\emptyset)+\sum_{n=1}^{\infty} \c_{n} \int_{E^{n}} f(x_{1},\cdots,x_{n})\, \p_{n}^{\mu}(x_{1},\cdots, x_{n})\dif x_{1}\ldots \dif x_{n}.
  \end{align*}
  The result follows by identification with \eqref{eq_preliminaries:1}.
\end{proof}
Then, Theorem \ref{thm:transport_fpp} applies as follows.
\begin{theorem}
  \label{thm:constructionTransport}
   Suppose that the hypothesis of Theorem  \ref{thm:egalite_vp} hold.
  Let
    $\Id-\nabla \psi_{n}$ be the optimal transport map between $\p^{\mu}_{n}$ and
    $\p^{\nu}_{n}$. Then, the optimal coupling is given by the following rule: For $\xi$ such that
    $\xi(E)=n$, it is coupled with $\zeta$ the configuration with $n$ atoms
    described by
    \begin{equation*}
      \zeta=\sum_{x\in \xi} \varepsilon_{x-\nabla_{x} \psi_{n}(\xi)}.
    \end{equation*}
    Furthermore,
    \begin{align*}
      \T_{c}(\mu,\nu)&=\sum_{n=1}^{\infty}\c_{n}\, \T_{e}(\p_{n}^{\mu},\, \p_{n}^{\nu})\\
      &=\frac{1}{2}\sum_{n=1}^{\infty}\int_{E^{n}}\|\nabla \psi_{n}(x_{1},\cdots,x_{n})\|^{2}_{E^{n}}\ j_{n}^{\mu}(x_{1},\cdots,x_{n})\dif x_{1}\ldots\dif x_{n}.
    \end{align*}
  \end{theorem}
  Theorems~\ref{thm:egalite_vp} and \ref{thm:constructionTransport} mean that
  two determinantal point processes are strongly coupled when and only when their eigenvalues
  are identical. Moreover, the eigenvalues also   control the convex combination
  of the densities of the projection DPP which  appear in the Janossy densities.

  \subsection{Determinantal projection processes}
\label{sec:projection}

Recall from Definition~\ref{def:projectionDPP} that a projection DPP has a
spectrum reduced to $\{1\}$. When it is of finite rank~$M$, almost all its
configurations have~$M$~points distributed according to the density
\begin{equation}\label{eq_transport:3}
 p_{\phi}(x_{1},\cdots,x_{M})=\frac{1}{M!} \det\Bigl(K_{\phi}(x_{i},\, x_{j}),\, 1\le i,j\le M \Bigr).
\end{equation}
Theorem~\ref{thm:constructionTransport} cannot be used as is since projection
DPPs do not possess Janossy densities. However, the initial definition of
$\T_{c}$ can still be used.
\begin{theorem}
  \label{thm_transport:distanceProjectionDPP}
  Let $\psi=(\psi_{j},\, 1\le j\le M)$ and $\psi=(\psi_{j},\, 1\le j\le M)$ two
  orthonormal families of $L^{2}(E;m)$. Let $\mu_{\psi}$ and $\mu_{\phi}$ the
  two projection DPP associated to these families. Then,
  \begin{equation*}
    \T_{c}(\mu_{\psi},\, \mu_{\phi})\le \inf_{\sigma\in \mathfrak S_{M}}\sum_{j=1}^{M} \T_{e}(|\psi_{j}|^{2}\dif m,\, |\phi_{\sigma(j)}|^{2}\dif m).
  \end{equation*}
\end{theorem}
\begin{proof}
We know that the points of $\mu_{\psi}$ (respectively $\mu_{\phi}$) are
distributed according to $p_{\psi}$ (respectively $p_{\phi}$) given by
\eqref{eq_transport:3}. Let $\gamma$ be a probability measure on $E^{M}\times
E^{M}$ whose marginals are $p_{\psi}\dif m$ and $p_{\phi}\dif m$. We know that
\begin{align*}
  \T_{c}(\mu_{\psi},\, \mu_{\phi})&=\int_{E^{M}\times E^{M}} \inf_{\sigma\in\mathfrak S_{M}}\sum_{j=1}^{M}|x_{j}-y_{\sigma(j)}|_{E}^{2}\ \dif \gamma(x_{1},\cdots,x_{M},y_{1},\cdots,y_{M})\\
                                  &\le \sum_{j=1}^{M}\int_{E^{M}\times E^{M}}|x_{j}-y_{j}|_{E}^{2}\ \dif \gamma(x_{1},\cdots,x_{M},y_{1},\cdots,y_{M}).
\end{align*}
We know from Algorithm 1, that the marginal distribution a single atom of
$\mu_{\psi}$ has distribution
\begin{equation*}
\dif \mu_{\psi}^{1}(x)=\frac1M  \sum_{j=1}^{M}|\psi_{j}(x)|^{2}\dif m(x).
\end{equation*}
Since $p_{\psi} $ and $p_{\phi}$ are both invariant with respect to
permutations, we obtain
\begin{align*}
  \T_{c}(\mu_{\psi},\, \mu_{\phi})&\le M\  \int_{E\times E} |x_{1}-y_{1}|_{E}^{2}\dif \gamma(x_{1},\cdots,x_{M},y_{1},\cdots,y_{M})\\
&\le M \ \T_{e}(\mu_{\psi}^{1},\mu_{\phi}^{1}).
  % &\le M\ \T_{e}\left(\frac1M  \sum_{j=1}^{M}|\psi_{j}(x)|^{2}\dif m(x),\ \frac1M  \sum_{j=1}^{M}|\phi_{j}(x)|^{2}\dif m(x)\right).
\end{align*}
If $\gamma_{i}^{1}$ is a coupling between $|\psi_{i}|^{2}\dif m$ and
$|\phi_{i}|^{2}\dif m$ then $M^{-1}\sum_{i=1}^{M}\gamma_{i}^{1}$ is a coupling
between $\mu_{\psi}^{1}$ and $\mu_{\phi}^{1}$. Hence,
\begin{align*}
  \T_{c}(\mu_{\psi},\, \mu_{\phi})&\le \sum_{j=1}^{M}\int_{E\times E}|x_{1}-y_{1}|^{2}\dif \gamma^{1}_{j}(x_{1},y_{1})\\
  &\le \sum_{j=1}^{M} \T_{e}(|\psi_{j}|^{2}\dif m,|\phi_{j}|^{2}\dif m).
\end{align*}
Since we can order the elements of the family $\psi$ and $\phi$ in any order,
the result follows.
\end{proof}

\section{Simulation}
\label{sec:simul}

As mentioned previously, implementing Algorithm~\ref{algo:sampling} using rejection sampling involves too many rejections which prevents the
algorithm to work for more than $1\,000$ points. In this section, we will
show that we can generate $10\,000$ points for the Ginibre point process
on a compact disc using inverse transform sampling and
approximation of the kernel.

In this section, we will consider Ginibre point processes but our reasoning
could be applied to any rotational invariant determinantal process like the
polyanalytic ensembles \cite{Haimi2011,Fenzel:aa} or the Bergman process
\cite{HoughZerosGaussiananalytic}.
For these processes, it is relatively easy to compute the eigenvalues and the
eigenfunctions of the kernel of their restriction to a ball centered at the
origin. For the Ginibre process, which will be our toy model, its restriction to
$\mathcal B_{R}$, denoted by $\Gin^{R}$, has a kernel of the form
\begin{equation*}
  K^{R}_{N}(x,y)=\sum_{n=0}^{\infty} \lambda_{n}^{R} \phi_{n}^{R}(x)\overline{\phi_{n}^{R}(y)}
\end{equation*}
where \cite{Decreusefond_2015}
\begin{align*}
  \lambda_{n}^{R}&=\frac{\gamma(n+1,R^{2})}{n!}\\
  \phi_{n}^{R}(x)&= \frac{1}{\sqrt{\pi \gamma(n+1,R^{2})}}\ x^{n} e^{-|x|^{2}/2},
\end{align*}
with $\gamma(n,r)$ is the lower incomplete gamma function.
We denote by the $\Gin^{R}_{N}$ the process whose kernel is the truncation of
$K_{R}$ to its first $N$ components:
\begin{equation*}
  K^{R}_{N}(x,y)=\sum_{n=0}^{N-1} \lambda_{n}^{R} \phi_{n}^{R}(x)\overline{\phi_{n}^{R}(y)}
\end{equation*}
The strict application of Algorithm~\ref{algo:sampling} for the simulation of
$\Gin^{R}$, requires to compute all
the quantities of the form
\begin{equation*}
  \lambda_{n}^{R} \prod_{k=n+1}^{\infty} (1-\lambda_{k}^{R})
\end{equation*}
to determine which Bernoulli random variables are \emph{active}. Strictly
speaking, this is unfeasible. However, it is a well
known observation that the number of points of $\Gin^{R}$ is about $R^{2}$. So
it is likely that $\Gin^{R}$ and $\Gin^{R}_{N_{R}}$ should be \emph{close} for
$N_{R}$ close to~$R^{2}$. This is what proves the next theorem.
\begin{theorem}
  \label{thm_simulation:troncatureeigenvalues}
  Let $c>0$ and $N_{R}=(R+c)^{2}$. For $R>c$, we have
  \begin{equation*}
    \T_{\KR}(\Gin^{R},\, \Gin^{R}_{N_{R}})\le \sqrt{\frac 2 \pi} R e^{-c^2}.
  \end{equation*}
\end{theorem}
Actually, the proof says that with high probability, $\Gin^{R}$ and
$\Gin_{N_{R}}^{R}$ do coincide.
\begin{proof}
  First, using the integral expression $\gamma(j,x) = \int_{t=0}^x
  t^{j-1}e^{-t}dt$, observe
  that $\sum_{j=1}^\infty \frac{\gamma(j, x)}{\Gamma(j)} = x$.
  Then, using the formula
  $\gamma(n+1, x) = n \gamma(n, x) - x^n e^{-x}$, we have by induction:
  \begin{align*}
    \sum_{j=n+1}^\infty \frac{\gamma(j, x)}{\Gamma(j)} = \frac{x^ne^{-x} - (n-x)\gamma(n, x)}{\Gamma(n)}\cdotp
  \end{align*}
  For $n=(R+c)^2$ and $x=R^2$, this implies:
  \begin{align*}
    \sum_{j\ge (R+c)^2}\lambda_j^R &\leq (R+c)^2\frac{R^{2(R+c)^2} e^{-R^2}}{ (R+c)^2 !}
  \end{align*}
  Using the bound $n! \geq \sqrt{2\pi n}\left(\frac n e\right)^n$
  \begin{align*}
    \sum_{j\ge (R+c)^2}\lambda_j^R &\leq \frac {R+c} {\sqrt{2\pi}}\frac{R^{2(R+c)^2} e^{(R+c)^2-R^2}}{ (R+c)^{2(R+c)^2}} \\
                         & \leq \frac {R+c} {\sqrt{2\pi}}e^{(R+c)^2-R^2 - 2(R+c)^2\log(1+\frac c R)}\\
                         & \leq \frac {R+c} {\sqrt{2\pi}}e^{(R+c)^2-R^2 - 2(R+c)^2 \frac{ \frac c R}{1+\frac c R}}\\
                         & \leq \frac {R+c} {\sqrt{2\pi}}e^{-c^2}
  \end{align*}
  Since  $R>c$, the proof is  complete.
\end{proof}
  As a corollary of the previous proof, we have
  \begin{equation}
    \label{eq:majoI}
    \P(\exists n\ge (R+c)^{2}, \text{Ber}(\lambda_{n}^{R})=1)\le
    \sum_{n\ge (R+c)^{2}} \lambda_{n}^{R}\le \kappa R e^{{-c^2}}
  \end{equation}
  for $R$ large enough. This means that the number of \emph{active} Bernoulli
  random variables in Algorithm~\ref{algo:sampling} is less than $(R+c)^{2}$
  with high probability.
  We can also provide a lower bound on the cardinality of $I$.
  \begin{lemma}\label{thm_simulation:lowerBound}
    For any $R > c > 0$,
    $$\P(\text{card}(I) < (R-c)^2)\le \frac 1 {\sqrt{2 \pi}}
    R e^{-c^2}.$$
  \end{lemma}
  \begin{proof}
    %Let $I^{*}=\max(j, \, j\in I)$. We have
    As in the previous proof, we will reduce the problem to bound a sum
    of reduced incomplete gamma functions.

    \begin{align*}
      \P(\text{card}(I) < (R-c)^2) &= 1 - \P(\text{card}(I) \geq (R-c)^2) \\
                                    &\leq 1 - \prod_{0 \leq j <\lfloor {(R-c)^{2}} \rfloor
                                    } \P(Ber(\lambda_j^R)=1) \\
                                    &\leq \sum_{0 \leq j < \lfloor {(R-c)^{2}} \rfloor}
                                    (1 - \P(Ber(\lambda_j^R)=1)) \\
                                    &\leq \sum_{1 \leq j \leq \lfloor {(R-c)^{2}} \rfloor}
                                    \frac{\Gamma(j,R^2)}{\Gamma(j)}\cdotp
                                    %&\leq R^2 \P(Y_j \leq (R-c)^2)
    \end{align*}
    Using the formula $\Gamma(n+1, x) = n \Gamma(n, x) + x^n e^{-x}$, we
    have by induction:
    \begin{align*}
      \sum_{j=1}^n \frac{\Gamma(j, x)}{\Gamma(j)} = \frac{x^ne^{-x} -
      (x-n)\Gamma(n, x)}{\Gamma(n)}\cdotp
    \end{align*}
    For $n=\lfloor {(R-c)^{2}} \rfloor$ and $x=R^2$, this implies:
    \begin{align*}
      \P(\text{card}(I) < (R-c)^2) &\leq (R-c)^2\frac{R^{2(R-c)^2} e^{-R^2}}{ (R-c)^2 !}\cdotp
    \end{align*}
    Using Stirling formula
    \begin{align*}
      \P(\text{card}(I) < (R-c)^2) &\leq \frac {R-c} {\sqrt{2\pi}}\frac{R^{2(R-c)^2} e^{(R-c)^2-R^2}}{ (R-c)^{2(R-c)^2}} \\
                           & \leq \frac {R-c} {\sqrt{2\pi}}e^{(R-c)^2-R^2 - 2(R-c)^2\log(1-\frac c R)}\\
                           & \leq \frac {R-c} {\sqrt{2\pi}}e^{(R-c)^2-R^2 + 2(R-c)^2 \frac{ \frac c R}{1-\frac c R}}\\
                           & \leq \frac {R-c} {\sqrt{2\pi}}e^{-c^2}\cdotp
    \end{align*}
The proof is thus complete.
\end{proof}
The combination of Lemma~\ref{thm_simulation:lowerBound} and \eqref{eq:majoI}
shows that the cardinality of $I$ is of the order of $R^{2}$ with high probability.
\subsection{Inverse transform sampling}
The next step of the algorithm is to draw the points according to a density
given by a determinant. Since we do not have explicit expression of the inverse
cumulative function of these densities, we have to resort to rejection sampling.
Fortunately, even it has not been noticed to the best of our knowledge, the
particular form of the eigenfunctions of the Ginibre like processes is prone to
the simulation of modulus and arguments by inverting their respective cumulative
distribution function.
This new approach is
summarized in Algorithm~\ref{algo:newsampling}.
\begin{lemma}[Simulation of the modules]
  Let $p(z) = \sum_{i \in I} a_i z^{n_i} f_i(|z|)$  and
  $$P(r) = \int_{\rho=0}^r\int_{\theta=0}^{2\pi}
  |p(\rho e^{j\theta})|^2 \rho\,d\rho\,d\theta$$
  The following equality holds:
  $$P(r) = \sum_{i \in I} |a_i|^2 F_i(r^2)$$
  where $$F_i(r^2) = \int_{\rho=0}^{r^2} \pi \rho^{n_i} f_i^2(\sqrt{\rho}) d\rho.$$
\end{lemma}

Given a sequence of complex numbers $W_\ell$ for $\ell$ from $1$ to
$|I|$, we denote by $e_\ell$ the orthonormal vectors obtained by
Gram-Schmidt orthonormalization of the vectors $\phi^{R}_I(W_\ell)$.
Let also $M_\ell \subset \mathbb R^{|I|}$ be the
vector $(|e_{\ell,i}|^2)_{i \in I}$ where $e_{\ell,i}$ is the coordinate of
index $i$ of $e_\ell$. Moreover, let $U_F(r) = (F_i(r))_{i \in I}$.
Finally, let $U_i$ be the sequence of vectors defined by induction with
$U_1=(1)_{i \in I}$ and $U_{i+1} = U_i - M_i$. Then drawing the module
of $W_i$ in Algorithm~\ref{algo:sampling} is reduced to sampling
uniformly $c_i$ in $[0,1]$ and solving the equation:
\begin{equation}
  c_i = \frac{1}{|I|-i+1} U_i \cdot U_F(r).
  \label{eq:module}
\end{equation}
Knowing $U_i$, we can compute $U_{i+1}$ in $O(|I|)$ arithmetic
operations. Using a dichotomy approach, Equation~\eqref{eq:module} can
be solved with precision $\varepsilon$ using $O(|I|\log \varepsilon)$
evaluation of the $F_i$.

Given the moduli, we can now simulate the arguments.
\begin{lemma}[Simulation of the arguments]
  Let $p = \sum_{i \in I} a_i z^{n_i} f_i(|z|)$ and
  $$Q(r,\alpha) = \int_{\theta=0}^{\alpha} |p(re^{i\theta})|^2 r\,d\theta$$
  Then $Q$ can be rewritten as a sum of $|I|^2$ terms:
  $$Q(r,\alpha) = \sum_{i, k \in I} a_i \overline{a_k} G_{i,k}(r,\alpha)$$
  where $G_{i,k}(r,\alpha) = \begin{cases}
    r g_i(r)g_k(r)\dfrac{e^{j(n_i-n_k)\alpha} -
    1}{j(n_i-n_k)} & \text{if } i\neq j\\
    r g_i^2(r)\alpha         & \text{if } i=k
  \end{cases}$ and $g_{i}(r) = r^{n_i} f_i(r)$.
\end{lemma}

Similarly to the simulation of the modules, for $\ell$ from $1$ to $|I|$,
let $A_\ell \subset \mathbb C^{|I|^2}$ be the vector
$(e_{\ell,i}\overline{e_{\ell,k}})_{i,k \in I}$. Let $V_G(r,\alpha) =
\left(G_{i,k}(r,\alpha)\right)_{i,k \in I}$. Let $(V_i)_{i=1 \ldots
|I|-1}$ be the sequence of vectors defined by recurrence with
$V_1=(\mathbbm 1_{i=k})_{i,k \in I}$ and
$V_{i+1} = V_i - A_i$.  Drawing the argument
of $W_i$ in Algorithm~\ref{algo:sampling} is now reduced to sampling
uniformly $c_i$ in $[0,1]$ and solving the equation:
\begin{equation}
  c_i = \frac{1}{V_i \cdot V_G(r, 2\pi)} V_i \cdot V_G(r,\alpha).
  \label{eq:argument}
\end{equation}

Computing $V_{i+1}$ from $V_i$ requires $O(|I|^2)$ arithmetic
operations. Then, for fixed~$r$, Equation~\eqref{eq:argument} can be
solved up to precision $\varepsilon$ in $O(|I|^2 + |I|\log \varepsilon)$ arithmetic operations and evaluations of the
$f_i$, using a dichotomy approach.

The total cost of sampling the $W_i$ with this approach is $O(|I|^3 +
|I|^2\log \varepsilon)$ operations. We will see in the next section how
we can reduce this complexity using an approximation of the eigenfunctions.

Gathering the results of this section, we get in
Algorithm~\ref{algo:newsampling} an efficient method to sample points
from a symmetric projection point process.

\SetKwComment{Comment}{}{}

\begin{algorithm}
  \KwData{$R,I$}
  \KwResult{$W_1,\cdots,W_{|I|}$}
  \BlankLine
  Draw $W_1$ from the distribution with density $\|\phi_I(x)\|_{\C^{|I|}}^2/|I|$\;
  $e_1\leftarrow \phi_I(W_1)/\|\phi_I(W_1)\|_{\C^{|I|}}$\;
  $U_1=(1)_{i \in I}$\;
  $V_1=(\mathbbm 1_{i=k})_{i,k \in I}$\;
  \For{$i \gets 2$ \KwTo $|I|$}{%
      \Comment{A. \em Update vectors $U_i$ and $V_i$ for next point simulation}
      $M_i \gets (|e_{i,\ell}|^2|)_{\ell \in I}$\;
      $A_i \gets (e_{i,k} \overline{e_{i,\ell}})_{k,\ell \in I}$\;
      $U_i \gets U_{i-1} - M_i$\;
      $V_i \gets V_{i-1} - A_i$\;
      \BlankLine
      \Comment{B. \em Draw point $W_i$}
      Draw $c_i$ from the uniform distribution in the interval $[0, 1]$\;
      $r_i \gets $ solution of $c_i = \frac{1}{|I|-i+1} U_i \cdot U_F(r)$\;
      Draw $d_i$ from the uniform distribution in the interval $[0, 1]$\;
      $\alpha_i \gets $ solution of $d_i = \frac{1}{V_i \cdot V_G(r_i, 2\pi)} V_i \cdot V_G(r_i,\alpha)$\;
      $W_i \gets r_i e^{i\alpha_i}$\;
      \BlankLine
      \Comment{C. \em Compute new vector $e_i$}
      $u_i\leftarrow \phi_I(W_i)-\sum_{k=1}^{i-1}e_k.\phi_I(W_i)\ e_k$\;
      $e_{i}\leftarrow u_i/\|u_i\|_{\C^{|I|}}$\;
  }
  \caption{Simulation of a compact symmetric projection point process
           restricted to the disc $\mathcal B_R$}
    \label{algo:newsampling}
\end{algorithm}

\subsection{Compact Ginibre and approximation}

Using Theorem~\ref{thm_transport:distanceProjectionDPP} with a well-chosen approximation, we
will show that we can reduce in Algorithm~\ref{algo:newsampling} the complexity
of steps \textsc{A.} and \textsc{B.} from $O(|I|^2)$ to $O(|I|^{1.5})$
operations with high probability.

For a given constant $c>0$ and for an integer $n$, let $R_n$ be the ring between
the circles of radii $u_{n} = \min(R,\sqrt{n}+c)$ and $l_n = \max(0,
\min(\sqrt{n}, R)-c)$. Let
\begin{align*}
  \mu_{n} =
  \int_{R_{n}} |\phi_{n}^{R}(z)|^2 dz &=
  \frac{\gamma(n+1, u_{n}^2) - \gamma(n+1, l_{n}^2)}{\gamma(n+1, R^2)}
  \\
  f_n(|z|) &= \frac{1}{\sqrt{\pi \gamma{(n+1,R^2)}}}e^{-\frac{|z|^2}{2}}.
\end{align*}
and define the following approximated functions :
$$\widetilde f_{n}(|z|) = \begin{cases}
  f_{n}(|z|)/\sqrt{\mu_{n}} & \text{if } z \in R_{n}\\
  0 & \text{otherwise}
\end{cases}$$ and let
\begin{equation*}
  \widetilde \phi_n^R(z) = z^{n} \widetilde f_{n}(|z|).
\end{equation*}
We now show that
replacing $\phi_n^R$ by $\tilde\phi_n^R$ does not cost much in terms of
accuracy.

\begin{theorem}
  For any $I\subset \{1,\cdots,N_{R}\}$,
  \begin{equation*}
    \T_{c}(\mu_{\phi},\, \mu_{\tilde\phi})\le \sum_{j\in I}\log\left( \frac{1}{\mu_{j}} \right)\cdotp
  \end{equation*}
\end{theorem}
\begin{proof}
  According to Theorem~\ref{thm_transport:distanceProjectionDPP}, it is
  sufficient to evaluate
  \begin{equation*}
    \T_{e}\left(|\phi_{j}^{R}|^{2}\dif x,\, |\tilde\phi_{j}^{R}|^{2}\dif x\right)
  \end{equation*}
  for any $j\in I$. Denote the two measures involved in the previous OTP by
  \begin{equation*}
    \zeta_{j}(dx)=|\phi_{j}^{R}(x)|^{2}\dif x,\ \tilde\zeta_{j}(dx)=|\tilde\phi_{j}^{R}(x)|^{2}\dif x.
  \end{equation*}
  These are two radially symmetric measures on $\R^{2}$. We still denote by
  $\zeta $ and $\tilde\zeta$ the two measures they induce on the polar
  coordinates $(r,\theta)$. Consider
  \begin{equation*}
    \zeta_{j}(\dif r \, |\, \theta)=c_{j}\,r^{2j+1}e^{-r^{2}}\car_{[0,R]}(r) \text{ where }c_{j}=\frac{1}{\gamma(j+1,R^{2})},
  \end{equation*}
  the distribution of $r$ given $\theta $ under $\zeta$ and the same quantity
  for $\tilde \zeta$. If we have a coupling $\Sigma_{\theta}$ between these two
  measures, then
  \begin{equation*}
    (r,\theta)\longmapsto (\Sigma_{\theta}(r),\theta)
  \end{equation*}
  is a coupling between $\zeta_{j}$ and $\tilde\zeta_{j}$. It follows that
  \begin{equation*}
    \T_{e}(\zeta_{j},\, \tilde \zeta_{j})\le \T_{e}\left(c_{j} r^{2j+1} e^{-r^{2}}\car_{[0,R]}(r)\dif r,\, \tilde c_{j} r^{2j+1} e^{-r^{2}}\car_{R_{j}}(r)\dif r\right)
  \end{equation*}
  where
  \begin{equation*}
    \tilde c_{j}=\frac{1}{\mu_{j}}\cdotp
  \end{equation*}
  We have
  \begin{equation*}
    -\frac{d^{2}}{dr^{2}}\log (r^{2j+1} e^{-r^{2}})=\frac{2j+1}{r^{2}}+2\ge 2.
  \end{equation*}
  Hence the Bakry-Emery criterion \cite{Villani2003} entails that the measure
  \begin{equation*}
   \rho_{\infty}(dr)=c_{j} r^{2j+1}
  e^{-r^{2}}\car_{[0,R]}(r)\dif r
  \end{equation*}
   satisfies the Talagrand inequality: For any
  probability measure $\rho$
  \begin{equation*}
    \T_{e}(\rho,\, \rho_{\infty})\le H(\rho \,| \, \rho_{\infty})= \int \rho(r)\log\frac{\rho(r)}{\rho_{\infty}(r)}\dif r
  \end{equation*}
  Apply this identity to
  \begin{equation*}
    \dif \rho_{j}(r)=\tilde c_{j} r^{2j+1} e^{-r^{2}}\car_{R_{j}}(r)\dif r
  \end{equation*}
  yields
  \begin{equation*}
     \T_{e}(\rho_{j},\, \rho_{\infty})\le \log\left( \frac{1}{\rho_{\infty}(R_{j})} \right)=\log\left( \frac{1}{\mu_{j}} \right)\cdotp
   \end{equation*}
   %% Voir Villani page 293
 %\marginpar{Y'a plus qu'à estimer ça...}
\end{proof}
Finally, using the same techniques as above,  we bound the sum of the $\log\left(\frac 1 {\mu_i}\right)$ in the following lemma.
\begin{lemma}
  There exists a constant $\kappa$ such that for $\sqrt{\log R} \leq c
  \leq R$:
  \begin{equation*}
    \sum_{n=0}^\infty\log\left( \frac{1}{\mu_{n}} \right) \leq \kappa R^2e^{-c^2}\cdotp
  \end{equation*}
\end{lemma}
\begin{proof}
We split the sum in three parts:
\begin{align*}
  S_1 &= \sum_{n=0}^{(R-c)^2-1}\log\left(\frac{1}{\mu_{n}}\right)\\
  S_2 &= \sum_{n=(R-c)^2}^{R^2-1}\log\left(\frac{1}{\mu_{n}}\right)\\
  S_3 &= \sum_{n=R^2}^\infty\log\left(\frac{1}{\mu_{n}}\right)\\
\end{align*}
We will first prove that the terms in $S_1$ and $S_2$ are $O(e^{-c^2})$
and the terms $\log(\frac 1 {\mu_{R^2+k}})$ in $S_3$ are
$O(Re^{-c^2}(1-\frac 1 R)^k)$.

For $c \geq 1$ and $n \leq (R-c)^2$, we show that ${\mu_n}^{-1}$ is roughly equal to
$\dfrac {\gamma(n+1, R^2)}{\Gamma(n+1)}$:
\begin{multline*}
  \gamma(n+1,  u_n^2) - \gamma(n+1, l_n^2)\\
  \begin{aligned}
&  = \Gamma(n+1) - \Gamma(n+1, u_n^2) - \gamma(n+1, l_n^2)\\
  &\geq \Gamma(n+1) - \frac{u_n^{2(n+1)}}{u_n^2 - n - 1}e^{-u_n^2}
            - \frac{l_n^{2(n+1)}}{n + 1 - l_n^2}e^{-l_n^2}\\
  &\geq \Gamma(n+1) - \frac{u_n^2 e^{-c^2}}{(u_n^2 - n - 1)\sqrt{2\pi n}}\Gamma(n+1)
     - \frac{l_n^2 e^{-c^2}}{(n + 1 - l_n^2)\sqrt{2\pi n}}\Gamma(n+1) \\
  &\geq \Gamma(n+1)(1 - e^{-c^2})
  \end{aligned}
\end{multline*}
This implies that
\begin{equation*}
 \frac 1 {\mu_n} \leq \frac{\gamma(n+1, R^2)}{\Gamma(n+1)(1 - e^{-c^2})}
\end{equation*}
and
\begin{equation*}
 \log(\frac 1 {\mu_n}) \leq \frac 1 {1 - e^{-c^2}} (\frac{\Gamma(n+1, R^2)}{\Gamma(n+1)} + e^{-c^2})\cdotp
\end{equation*}
Thus,  for $c \geq 1$, we have:
\begin{equation*}
  S_1 \leq 2Re^{-c^2} + 2R^2e^{-c^2}.
\end{equation*}
For $S_2$ we have $u_n = R$ and $l_n = \sqrt n - c$ such that
\begin{equation*}
 \frac 1
{\mu_n} = \frac 1 {1 - \dfrac {\gamma(n+1, (\sqrt n -c)^2)}{\gamma(n+1,
R^2)}}\cdotp
\end{equation*}
Moreover, for $n + 1 \leq R^2$, we know that
\begin{align*}
 \gamma(n+1, R^2)
  &\geq \frac {\Gamma(n+1)} 2 \\
  \intertext{ and }
  \gamma(n+1, (\sqrt n -c)^2) &\leq e^{-c^2}
\Gamma(n+1).
\end{align*}
Combine these identities with the well known fact
\begin{equation*}
  \sum_{i} \log \left(\frac 1 {1-\epsilon_i}\right)
\leq \frac{\sum_{i} \epsilon_i}{1-\max \epsilon_i}
\end{equation*}
to obtain
\begin{equation*}
  S_2 \leq \frac{2cR e^{-c^2}}{1 - e^{-c^2}}\cdotp
\end{equation*}
Finally for $S_3$, we have $u_n = R$ and $l_n = R-c$, so that we have
\begin{equation*}
  \frac 1 {\mu_n}
= \frac 1 {1 - \frac {\gamma(n+1, (R-c)^2)}{\gamma(n+1, R^2)}}\cdotp
\end{equation*}
 Then remark that :
\begin{align*}
  \frac {\gamma(n, (R-c)^2)}{\gamma(n, R^2)}
  &\leq \frac{(R-c)^{2n}e^{-(R-c)^2}/(n-(R-c)^2)}{R^{2n}e^{-R^2}/n} \\
  &\leq \frac{(1-\frac c R)^{2n}e^{R^2-(R-c)^2}}{1-\frac{(R-c)^2} n} \\
  &\leq \frac{(1-\frac c R)^{2(n-R^2)}e^{- 2R^2\frac c R - R^2\frac{c^2}{R^2} + 2cR - c^2}}{1-\frac{(R-c)^2} n} \\
  &\leq \frac{(1-\frac c R)^{2(n-R^2)}e^{-2c^2}}{1-\frac{(R-c)^2}{R^2}} \\
  &\leq \frac R c (1-\frac c R)^{2(n-R^2)}e^{-2c^2} \\
\end{align*}
Thus, summing from $R^2$ to $\infty$ we get:
\begin{equation*}
  S_3 \leq \frac{\dfrac {R^2} {c^2} e^{-2c^2}}{1-\dfrac {R} {c}
  e^{-2c^2}}\cdotp
\end{equation*}
The proof is thus complete.

\end{proof}

\subsection{Experimental results}

An implementation in Python of this algorithm publicly available
\cite{Mzenodo20} allowed us to sample $10\,000$ points in $2\,128$ seconds on a
8 core 3Ghz CPU. The same approach can be used for the DPP with the so-called
Bergmann kernel which represents the zeros of some Gaussian analytic
functions~\cite{HoughZerosGaussiananalytic}.

% \bibliographystyle{plain}
% \bibliography{biblioDPP,references}

\end{document}